\documentclass[a4paper,UKenglish]{lipics}
\usepackage{tikz}
  \usetikzlibrary{arrows,automata,shapes}
  \usetikzlibrary{backgrounds}
\usepackage{microtype}

\usepackage{todonotes}

\newcommand{\fun}[1]{\ensuremath{\textsl{#1}}}
\newcommand{\eps}{\varepsilon}

\DeclareMathOperator{\alp}{\rm alph}
\DeclareMathOperator{\A}{\mathcal A}
\DeclareMathOperator{\B}{\mathcal B}
\DeclareMathOperator{\D}{\mathcal D}

\DeclareMathOperator{\M}{\mathcal M}
\theoremstyle{plain}
\newtheorem{fact}[theorem]{Fact}

\makeatletter
\newcommand{\newreptheorem}[2]{\newtheorem*{rep@#1}{\rep@title}\newenvironment{rep#1}[1]{\def\rep@title{#2 \ref*{##1}}\begin{rep@#1}}{\end{rep@#1}}}
\makeatother
\newreptheorem{theorem}{Theorem}
\newreptheorem{lemma}{Lemma}
\newreptheorem{corollary}{Corollary}

\title{Piecewise Testable Languages and Nondeterministic Automata\footnote{This work was supported by the DFG in project DIAMOND (Emmy Noether grant KR~4381/1-1).}}

\author{Tom\'{a}\v{s} Masopust}
\affil{Fakult\"at Informatik, Technische Universit\"at Dresden, Germany and\\
    Institute of Mathematics CAS, Czech Republic\\
  \texttt{tomas.masopust@tu-dresden.de}}
\authorrunning{T. Masopust} 

\Copyright{Tom\'{a}\v{s} Masopust}


\serieslogo{}
\volumeinfo
  {Billy Editor and Bill Editors}
  {2}
  {Conference title on which this volume is based on}
  {1}
  {1}
  {1}
\EventShortName{ArXiv}
\DOI{}

\begin{document}

\maketitle

\begin{abstract}
  A regular language is $k$-piecewise testable if it is a finite boolean combination of languages of the form $\Sigma^* a_1 \Sigma^* \cdots \Sigma^* a_n \Sigma^*$, where $a_i\in\Sigma$ and $0\le n \le k$. Given a DFA $\A$ and $k\ge 0$, it is an NL-complete problem to decide whether the language $L(\A)$ is piecewise testable and, for $k\ge 4$, it is coNP-complete to decide whether the language $L(\A)$ is $k$-piecewise testable. It is known that the depth of the minimal DFA serves as an upper bound on $k$. Namely, if $L(\A)$ is piecewise testable, then it is $k$-piecewise testable for $k$ equal to the depth of $\A$. In this paper, we show that some form of nondeterminism does not violate this upper bound result. Specifically, we define a class of NFAs, called ptNFAs, that recognize piecewise testable languages and show that the depth of a ptNFA provides an (up to exponentially better) upper bound on $k$ than the minimal DFA. We provide an application of our result, discuss the relationship between $k$-piecewise testability and the depth of NFAs, and study the complexity of $k$-piecewise testability for ptNFAs.
\end{abstract}

\section{Introduction}
  A regular language $L$ over an alphabet $\Sigma$ is {\em piecewise testable\/} if it is a finite boolean combination of languages of the form 
  \[
    L_{a_1a_2\ldots a_n} = \Sigma^* a_1 \Sigma^* a_2 \Sigma^* \cdots \Sigma^* a_n \Sigma^*
  \]
  where $a_i\in \Sigma$ and $n\ge 0$. If $L$ is piecewise testable, then there exists a nonnegative integer $k$ such that $L$ is a finite boolean combination of languages $L_u$, where the length of $u\in\Sigma^*$ is at most $k$. In this case, the language $L$ is called {\em $k$-piecewise testable}.

  Piecewise testable languages are studied in semigroup theory~\cite{Almeida2008486,AlmeidaZ-ita97,PerrinPin} and in logic over words~\cite{DiekertGK08,mfcsPlaceRZ13} because of their close relation to first-order logic $\textrm{FO}(<)$. They actually form the first level of the Straubing-Th\'erien hierarchy~\cite{PerrinP86,Thomas82}. This hierarchy is closely related to the dot-depth hierarchy~\cite{CohenB71}, see more in~\cite{KufleitnerL12}. They are indeed studied in formal languages and automata theory~\cite{KlimaP13}, recently mainly in the context of separation~\cite{mfcsPlaceRZ13,lvanrooijen}. Although the separability of context-free languages by regular languages is undecidable, separability by piecewise testable languages is decidable~\cite{CzerwinskiM14} (even for some non-context-free families). Piecewise testable languages form a strict subclass of star-free languages, that is, of the limit of the above-mentioned hierarchies or, in other words, of the languages definable by LTL logic. They are investigated in natural language processing~\cite{FuHT2011,Rogers:2007}, in cognitive and sub-regular complexity~\cite{RogersHFHLW12}, in learning theory~\cite{GarciaR04,Kontorovich2008}, and in databases in the context of XML schema languages~\cite{icalp2013,HofmanM15,mfcs2014}. They have been extended from words to trees~\cite{Bojanczyk:2012,GarciaV90}.
  
  Recently, the complexity of computing the minimal $k$ and/or bounds on $k$ for which a piecewise testable language is $k$-piecewise testable was studied in~\cite{HofmanM15,KKP,KlimaP13}, motivated by applications in databases and algebra and logic. However, the knowledge of such a $k$ that is either minimal or of reasonable size is of interest in many other applications as well. The complexity to test whether a piecewise testable language is $k$-piecewise testable was shown to be coNP-complete for $k\ge 4$ if the language is given as a DFA~\cite{KKP} and PSPACE-complete if the language is given as an NFA~\cite{dlt15}. The complexity for DFAs and $k<4$ is discussed in detail in~\cite{dlt15}. The best upper bound on $k$ known so far is given by the depth of the minimal DFA~\cite{KlimaP13}.
  
  In this paper, we define a class of NFAs, called ptNFAs, that characterizes piecewise testable languages. This characterization is based on purely structural properties, therefore it is NL-complete to check whether an NFA is a ptNFA (Theorem~\ref{thmMainNL}). We show that the depth of ptNFAs also provides an upper bound on $k$-piecewise testability (Theorem~\ref{thmMain}) and that this new bound is up to exponentially lower than the one given by minimal DFAs (Section~\ref{demo} and Theorem~\ref{thmApp}). We further show that this property does not hold for general NFAs, and that the gap between $k$-piecewise testability and the depth of NFAs can be arbitrarily large (Theorem~\ref{propMain}). The opposite implication of Theorem~\ref{thmMain} does not hold. We give a non-trivial application of our result in Section~\ref{sec5}, where we also provide more discussion. Finally, in Section~\ref{sec4}, we discuss the complexity of $k$-piecewise testability for ptNFAs.
  
  The paper is organized as follows. Section~\ref{sec2} presents basic notions and definitions, fixes the notation, and defines the ptNFAs. Section~\ref{sec3} motivates and demonstrates Theorem~\ref{thmMain} on a simple example. Section~\ref{secDepth} then proves Theorem~\ref{thmMain} and the related results. Section~\ref{sec5} provides a non-trivial application and further discussion. Section~\ref{sec4} recalls the known complexity results and discusses the complexity of the related problems for ptNFAs. Section~\ref{conclusion} concludes the paper.

\section{Preliminaries and Definitions}\label{sec2}
  We assume that the reader is familiar with automata theory, see, e.g., \cite{AhoHU74}. The cardinality of a set $A$ is denoted by $|A|$ and the power set of $A$ by $2^A$. An alphabet, $\Sigma$, is a finite nonempty set; the elements of an alphabet are called symbols or letters. The free monoid generated by $\Sigma$ is denoted by $\Sigma^*$. A word over $\Sigma$ is any element of $\Sigma^*$; the empty word is denoted by $\eps$. For a word $w\in\Sigma^*$, $\alp(w)\subseteq\Sigma$ denotes the set of all letters occurring in $w$, and $|w|_a$ denotes the number of occurrences of letter $a$ in $w$. A language over $\Sigma$ is a subset of $\Sigma^*$. For a language $L$ over $\Sigma$, let $\overline{L}=\Sigma^*\setminus L$ denote the complement of $L$.

  A {\em nondeterministic finite automaton\/} (NFA) is a quintuple $\A = (Q,\Sigma,\cdot,I,F)$, where $Q$ is a finite nonempty set of states, $\Sigma$ is an input alphabet, $I\subseteq Q$ is a set of initial states, $F\subseteq Q$ is a set of accepting states, and $\cdot : Q\times\Sigma \to 2^Q$ is the transition function that can be extended to the domain $2^Q \times \Sigma^*$ by induction. The language {\em accepted\/} by $\A$ is the set $L(\A) = \{w\in\Sigma^* \mid I \cdot w \cap F \neq \emptyset\}$. In what follows, we usually omit $\cdot$ and write simply $Iw$ instead of $I\cdot w$. 
  
  A {\em path\/} $\pi$ from a state $q_0$ to a state $q_n$ under a word $a_1a_2\cdots a_{n}$, for some $n\ge 0$, is a sequence of states and input symbols $q_0 a_1 q_1 a_2 \ldots q_{n-1} a_{n} q_n$ such that $q_{i+1} \in q_i\cdot a_{i+1}$, for all $i=0,1,\ldots,n-1$. The path $\pi$ is {\em accepting\/} if $q_0\in I$ and $q_n\in F$. We use the notation $q_0 \xrightarrow{a_1a_2\cdots a_{n}} q_{n}$ to denote that there exists a path from $q_0$ to $q_n$ under the word $a_1a_2\cdots a_{n}$. A path is {\em simple\/} if all states of the path are pairwise distinct. The number of states on the longest simple path of $\A$, starting in an initial state, decreased by one (i.e., the number of transitions on that path) is called the {\em depth\/} of the automaton $\A$, denoted by $\fun{depth}(\A)$. 
  
  The NFA $\A$ is {\em complete\/} if for every state $q$ of $\A$ and every letter $a$ in $\Sigma$, the set $q\cdot a$ is nonempty, that is, in every state, a transition under every letter is defined.
  
  Let $\A=(Q,\Sigma,\cdot,I,F)$ be an NFA, and let $p$ be a state of $\A$. The sub-automaton of $\A$ induced by state $p$ is the automaton $\A_p = (\fun{reach}(p),\Sigma,\cdot_p,p,F\cap \fun{reach}(p))$ with state $p$ being the sole initial state and with only those states of $\A$ that are reachable from $p$; formally, $\fun{reach}(p)$ denotes the set of all states reachable from state $p$ in $\A$ and $\cdot_p$ is a restriction of $\cdot$ to $\fun{reach}(p)\times\Sigma$.

  The NFA $\A$ is {\em deterministic\/} (DFA) if $|I|=1$ and $|q\cdot a|=1$ for every state $q$ in $Q$ and every letter $a$ in $\Sigma$. Then the transition function $\cdot$ is a map from $Q\times\Sigma$ to $Q$ that can be extended to the domain $Q\times\Sigma^*$ by induction. Two states of a DFA are {\em distinguishable\/} if there exists a word $w$ that is accepted from one of them and rejected from the other. A DFA is {\em minimal\/} if all its states are reachable and pairwise distinguishable. 

  Let $\A=(Q,\Sigma,\cdot,I,F)$ be an NFA. The reachability relation $\le$ on the set of states is defined by $p\le q$ if there exists a word $w$ in $\Sigma^*$ such that $q\in p\cdot w$. The NFA $\A$ is {\em partially ordered\/} if the reachability relation $\le$ is a partial order. For two states $p$ and $q$ of $\A$, we write $p < q$ if $p\le q$ and $p\ne q$. A state $p$ is {\em maximal\/} if there is no state $q$ such that $p < q$. Partially ordered automata are sometimes also called acyclic automata.
  
  An NFA $\A=(Q,\Sigma,\cdot,I,F)$ can be turned into a directed graph $G(\A)$ with the set of vertices $Q$, where a pair $(p,q)$ in $Q \times Q$ is an edge in $G(\A)$ if there is a transition from $p$ to $q$ in $\A$. For $\Gamma \subseteq \Sigma$, we define the directed graph $G(\A,\Gamma)$ with the set of vertices $Q$ by considering all those transitions that correspond to letters in $\Gamma$. For a state $p$, let $\Sigma(p)=\{a\in\Sigma \mid p\in p\cdot a\}$ denote the set of all letters under which the NFA $\A$ has a self-loop in state $p$. Let $\A$ be a partially ordered NFA. If for every state $p$ of $\A$, state $p$ is the unique maximal state of the connected component of $G(\A,\Sigma(p))$ containing $p$, then we say that the NFA satisfies the {\em unique maximal state (UMS) property}.
  
  An equivalent notion to the UMS property for minimal DFAs has been introduced in the literature. A DFA $\A$ over $\Sigma$ is {\em confluent\/} if, for every state $q$ of $\A$ and every pair of letters $a, b$ in $\Sigma$, there exists a word $w$ in $\{a, b\}^*$ such that $(q a) w = (q b) w$. 
  
  We adopt the notation $L_{a_1 a_2 \cdots a_n} = \Sigma^* a_1 \Sigma^* a_2 \Sigma^* \cdots \Sigma^* a_n \Sigma^*$ from~\cite{KlimaP13}. For two words $v = a_1 a_2 \cdots a_n$ and $w \in L_{v}$, we say that $v$ is a {\em subsequence\/} of $w$ or that $v$ can be {\em embedded\/} into $w$, denoted by $v \preccurlyeq w$. For $k\ge 0$, let $\fun{sub}_k(v) =\{u\in\Sigma^* \mid u\preccurlyeq v,\, |u|\le k\}$. For two words $w_1,w_2$, define $w_1 \sim_k w_2$ if and only if $\fun{sub}_k(w_1)=\fun{sub}_k(w_2)$. Note that $\sim_k$ is a congruence with finite index.
  
  The following is well known.
  \begin{fact}[\cite{Simon1972}]\label{mainProperty}
    Let $L$ be a regular language, and let $\sim_L$ denote the Myhill congruence~\cite{Myhill}. A language $L$ is $k$-piecewise testable if and only if $\sim_k\subseteq \sim_L$. Moreover, $L$ is a finite union of $\sim_k$ classes. 
  \end{fact}
  
  We will use this fact in several proofs in the form that if $L$ is not $k$-piecewise testable, then there exist two words $u$ and $v$ such that $u \sim_k v$ and $|L\cap\{u,v\}|=1$.
  
  \begin{fact}\label{thm:characterization}
    Let $L$ be a language recognized by the minimal DFA $\A$. The following is equivalent.
    \begin{enumerate}
      \item The language $L$ is piecewise testable.
      \item The minimal DFA $\A$ is partially ordered and confluent~\cite{KlimaP13}.
      \item The minimal DFA $\A$ is partially ordered and satisfies the UMS property~\cite{Trahtman2001}.
    \end{enumerate}
  \end{fact}

  We now define a special class of nondeterministic automata called ptNFAs. The name comes from piecewise testable, since, as we show below, they characterize piecewise testable languages. And indeed include all minimal DFAs recognizing piecewise testable languages.
  \begin{definition}\label{ptNFA}
    An NFA $\A$ is called a {\em ptNFA\/} if it is partially ordered, complete and satisfies the UMS property.
  \end{definition}

  The reason why we use the UMS property in the definition of ptNFAs rather than the notion of confluence is simply because confluence does not naturally generalize to NFAs, as shown in Example~\ref{exUMSnotConf} below.
  \begin{example}\label{exUMSnotConf}
    Consider the automaton depicted in Figure~\ref{fig1}. The notion of confluence is not clear for NFAs. If we consider the point of view that whenever the computation is split, a common state can be reached under a word over the splitting alphabet, then this automaton is confluent. However, it does not satisfy the UMS property and its language is not piecewise testable; there is an infinite sequence $a,ab,aba,abab,\ldots$ that alternates between accepted and non-accepted words, which implies that there is a non-trivial cycle in the corresponding minimal DFA and, thus, it proves non-piecewise testability by Fact~\ref{thm:characterization}.
  \end{example}

  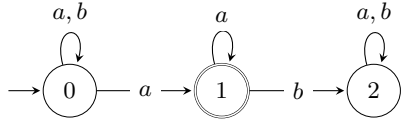
\begin{figure}
    \centering
    \begin{tikzpicture}[baseline,->,>=stealth,shorten >=1pt,node distance=2cm,
      state/.style={circle,minimum size=7mm,very thin,draw=black,initial text=},
      every node/.style={fill=white,font=\small},
      bigloop/.style={shift={(0,0.01)},text width=1.6cm,align=center}]
      \node[state,initial]    (0) {$0$};
      \node[state,accepting]  (1) [right of=0] {$1$};
      \node[state]            (2) [right of=1] {$2$};
      \path
        (0) edge node {$a$} (1)
        (1) edge node {$b$} (2)
        (0) edge[loop above] node[bigloop] {$a, b$} (0)
        (1) edge[loop above] node[bigloop] {$a$} (1)
        (2) edge[loop above] node[bigloop] {$a, b$} (2)
        ;
    \end{tikzpicture}
    \caption{Confluent automaton accepting a non-piecewise testable language}
    \label{fig1}
  \end{figure}

  Note that to check whether an NFA is a ptNFA requires to check whether the automaton is partially ordered, complete and satisfies the UMS property. The violation of these properties can be tested by several reachability tests, hence its complexity belongs to coNL=NL. On the other hand, to check the properties is NL-hard even for minimal DFAs~\cite{ChoH91}. Thus, we have the following.
  
  \begin{theorem}\label{thmMainNL}
    It is NL-complete to check whether an NFA is a ptNFA.
  \end{theorem}

\section{Motivation and an Example}\label{demo}\label{sec3} 
  Considering applications, such as XML, where the alphabet can hardly be considered as fixed, the results of~\cite{KKP} (cf. Theorem~\ref{thmconp} below) say that it is intractable to compute the minimal $k$ for which a piecewise testable language is $k$-piecewise testable, unless coNP=P. This leads to the investigation of reasonably small upper bounds. Recall that the result of~\cite{KlimaP13} says that $k$ is bounded by the depth of the minimal DFA. However, applications usually require to work with NFAs, which motivates the research of this paper. Another motivation comes from a simple observation that, given several DFAs, a result of an operation can lead to an NFA that in some sense still have the DFA-like properties, see more discussion below. Moreover, it seems to be a human nature to use a kind of nondeterminism, for instance, to reuse already defined parts as demonstrated here on a very simple example.

  Let $L_0=\{\eps\}$ be a language over the alphabet $\Sigma_0=\{a_0\}$. Assume that the language $L_i$ over $\Sigma_i$ is defined, and let $L_{i+1} = L_i \cup \Sigma_i^* a_{i+1} L_i$ over $\Sigma_{i+1}=\Sigma_i\cup\{a_{i+1}\}$, where $a_{i+1}$ is a new symbol not in $\Sigma_i$. We now construct the NFAs for the languages $L_i$,
  \[
    \A_i=(\{0,1,\ldots,i\},\{a_0,a_1,\ldots,a_i\},\cdot,\{0,1,\ldots,i\},\{0\})
  \]
  where $\ell \cdot a_j = \ell$ if $i\ge \ell > j \ge 0$ and $\ell \cdot a_\ell = \{0,1,\ldots,\ell-1\}$ if $i\ge \ell\ge 1$. The automaton $\A_3$ is depicted in Figure~\ref{fig6}. The dotted transitions are to ``complete'' the NFA in the meaning that $\ell\cdot a \neq \emptyset$ for any state $\ell$ and letter $a$.
  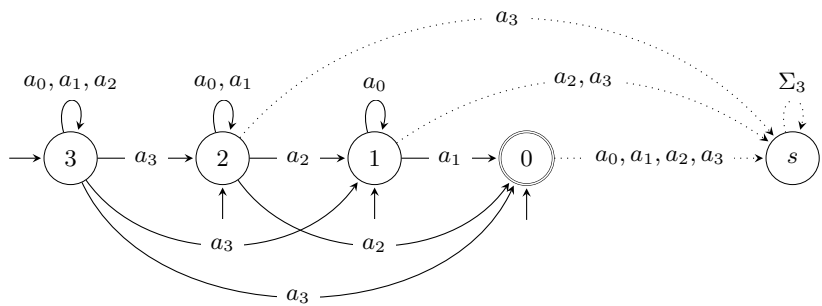
\begin{figure}
    \centering
    \begin{tikzpicture}[baseline,->,>=stealth,shorten >=1pt,node distance=2cm,
      state/.style={circle,minimum size=7mm,very thin,draw=black,initial text=},
      every node/.style={fill=white,font=\small},
      bigloop/.style={shift={(0,0.01)},text width=1.6cm,align=center}]
      \node[state,initial below,accepting]    (0) {$0$};
      \node[state,initial below]              (1) [left of=0] {$1$};
      \node[state,initial below]              (2) [left of=1] {$2$};
      \node[state,initial]                    (3) [left of=2] {$3$};
      \node[state]                            (s) [right of=0,node distance=3.5cm] {$s$};
      \path
        (3) edge node {$a_3$} (2)
        (2) edge node {$a_2$} (1)
        (1) edge node {$a_1$} (0)
        (3) edge[bend right=55] node {$a_3$} (1)
        (3) edge[bend right=65] node {$a_3$} (0)
        (2) edge[bend right=55] node {$a_2$} (0)
        (3) edge[loop above] node[bigloop] {$a_0, a_1, a_2$} (3)
        (2) edge[loop above] node[bigloop] {$a_0, a_1$} (2)
        (1) edge[loop above] node[bigloop] {$a_0$} (1)
        (2) edge[dotted,bend left=50] node {$a_3$} (s)
        (1) edge[dotted,bend left=35] node {$a_2,a_3$} (s)
        (0) edge[dotted] node {$a_0,a_1,a_2,a_3$} (s)
        (s) edge[dotted,loop above] node {$\Sigma_3$} (s)
        ;
    \end{tikzpicture}
    \caption{Automaton $\A_3$; the dotted transitions depict the completion of $\A_3$}
    \label{fig6}
  \end{figure}
    
  Although the example is very simple, the reader can see the point of the construction in nondeterministically reusing the existing parts. 
    
  Now, to decide whether the language is piecewise testable and, if so, to obtain an upper bound on its $k$-piecewise testability, the known results for DFAs say that we need to compute the minimal DFA. Doing so shows that $L_i$ is piecewise testable. However, the minimal DFA for the language $L_i$ is of exponential size and its depth is $2^{i+1}-1$, cf.~\cite{dlt15}, which implies that $L_i$ is $(2^{i+1}-1)$-piecewise testable. Another way is to use the PSPACE algorithm of~\cite{dlt15} to compute the minimal $k$. Both approaches are basically of the same complexity.
    
  This is the place, where our result comes into the picture. According to Theorem~\ref{thmMain} proved in the next section, the easily testable structural properties say that the language $L_i$ is $(i+1)$-piecewise testable. This provides an exponentially better upper bound for every language $L_i$ than the technique based on minimal DFAs. Finally, we note that it can be shown that $L_i$ is not $i$-piecewise testable, so the bound is tight.

\section{Piecewise Testability and Nondeterminism}\label{secDepth}
  In this section, we establish a relation between piecewise testable languages and nondeterministic automata and generalize the bound given by the depth of DFAs to ptNFAs. We first recall the know result for DFAs.
  
  \begin{theorem}[\cite{KlimaP13}]\label{thmDFAs}
    Let $\A$ be a partially ordered and confluent DFA. If the depth of $\A$ is $k$, then the language $L(\A)$ is $k$-piecewise testable.
  \end{theorem}
  
  This result is currently the best known structural upper bound on $k$-piecewise testability. The opposite implication of the theorem does not hold and we have shown in~\cite{dlt15} (see also Section~\ref{demo}) that this bound can be exponentially far from the minimal value of $k$. This observation has motivated our investigation of the relationship between piecewise testability and the depth of NFAs. We have already generalized a structural automata characterization for piecewise testability from DFAs to NFAs as follows.
  
  \begin{theorem}[\cite{dlt15}]\label{thm10}
    A regular language is piecewise testable if and only if it is recognized by a ptNFA.
  \end{theorem}
  
  We now generalize Theorem~\ref{thmDFAs} to ptNFAs and discuss the relation between the depth of NFAs and $k$-piecewise testability in more detail. An informal idea behind the proof is that every ptNFA can be ``decomposed'' into a finite number of partially ordered and confluent DFAs. We now formally prove the theorem by generalizing the proof of Theorem~\ref{thmDFAs} given in~\cite{KlimaP13}.
  
  \begin{theorem}\label{thmMain}
    If the depth of a ptNFA $\A$ is $k$, then the language $L(\A)$ is $k$-piecewise testable.
  \end{theorem}
  
  The proof of Theorem~\ref{thmMain} follows directly from Lemmas~\ref{lemmaAi} and~\ref{lemmaMain} proved below.
  
  \begin{lemma}\label{lemmaAi}
    Let $\A$ be a ptNFA with $I$ denoting the set of initial states. Then the language $L(\A)=\bigcup_{i\in I} L(\A_i)$, where every sub-automaton $\A_i$ is a ptNFA.
  \end{lemma}
  
  Based on the previous lemma, it is sufficient to show the theorem for ptNFAs with a single initial state. We make use of the following lemma.
  \begin{lemma}[\cite{KlimaP13}]\label{lemma2}
    Let $\ell \ge 1$, and let $u, v \in \Sigma^*$ be such that $u \sim_\ell v$. Let $u = u'au''$ and $v = v'av''$ such that $a\notin \alp(u'v')$. Then $u'' \sim_{\ell-1} v''$.
  \end{lemma}  
  
  \begin{lemma}\label{lemmaMain}
    Let $\A$ be a ptNFA with a single initial state and depth $k$. Then the language $L(\A)$ is $k$-piecewise testable.
  \end{lemma}
  \begin{proof}
    Let $\A=(Q,\Sigma,\cdot,i,F)$. If the depth of $\A$ is 0, then $L(\A)$ is either $\emptyset$ or $\Sigma^*$, which are both $0$-piecewise testable by definition. Thus, assume that the depth of $\A$ is $\ell \ge 1$ and that the claim holds for ptNFAs of depth less than $\ell$. Let $u, v\in \Sigma^*$ be such that $u \sim_{\ell} v$. We prove that $u$ is accepted by $\A$ if and only if $v$ is accepted by $\A$.
    
    Assume that $u$ is accepted by $\A$ and fix an accepting path of $u$ in $\A$. If $\alp(u)\subseteq \Sigma(i)$, then the UMS property of $\A$ implies that $i\in F$. Therefore, $v$ is also accepted in $i$. If $\alp(u)\not\subseteq\Sigma(i)$, then $u = u'au''$ and $v = v'bv''$, where $u',v'\in\Sigma(i)^*$, $a,b\in \Sigma\setminus\Sigma(i)$, and $u'',v''\in \Sigma^*$. Let $p\in i\cdot a$ be a state on the fixed accepting path of $u$. Let $\A_p=(\fun{reach}(p),\Sigma,\cdot_p,p,F\cap \fun{reach}(p))$ be a sub-automaton of $\A$ induced by state $p$. Note that $\A_p$ is a ptNFA. By assumption, $\A_p$ accepts $u''$ and the depth of $\A_p$ is at most $\ell-1$.

    If $a = b$, Lemma~\ref{lemma2} implies that $u'' \sim_{\ell-1} v''$. By the induction hypothesis, $u''$ is accepted by $\A_p$ if and only if $v''$ is accepted by $\A_p$. Hence, $v=v'av''$ is accepted by $\A$.

    If $a\neq b$, then $u = u'au_0''bu_1''$ and $v = v'bv_0''av_1''$, where $b \notin \alp(u'au_0'')$ and $a \notin \alp(v'bv_0'')$. Then 
    \[
      u'' = u_0''bu_1'' \sim_{\ell-1} v_0''av_1'' = v''
    \]
    because, by Lemma~\ref{lemma2}, 
    \begin{align}\tag{*}\label{eq01}
      \fun{sub}_{\ell-1}(u_0'' b u_1'') = \fun{sub}_{\ell-1} (v_1'') \subseteq \fun{sub}_{\ell-1} (v_0'' av_1'') = \fun{sub}_{\ell-1} (u_1'') \subseteq \fun{sub}_{\ell-1} (u_0'' bu_1'')\,. 
    \end{align}
    
    If $p\in i\cdot b$, the induction hypothesis implies that $v''$ is accepted by $\A_p$, hence $v=v'bv''$ is accepted by $\A$. 
    
    If $p\notin i\cdot b$, let $q\in i\cdot b$. By the properties of $\A$, there exists a word $w \in \{a, b\}^*$ such that $pw = qw = r$, for some state $r$. Indeed, there exists $w_1$ and a unique maximal state $r$ with respect to $\{a,b\}$ such that $p w_1 = \{r\}$ and $a,b\in\Sigma(r)$. By the UMS property, there exists $w_2$ such that $q w_1w_2=\{r\}$. Let $w=w_1w_2$. We now show that $wu'' \sim_{\ell-1} u''$ by induction on the length of $w$. There is nothing to show for $w=\eps$. Thus, assume that $w=xw'$, for $x\in\{a,b\}$, and that $w'u'' \sim_{\ell-1} u''$.
    Notice that (\ref{eq01}) shows that $u'' \sim_{\ell-1} v_1'' \sim_{\ell-1} v'' \sim_{\ell-1} u_1''$.
    This implies that $\fun{sub}_{\ell-1}(v_1'') \subseteq \fun{sub}_{\ell-1}(av_1'') \subseteq \fun{sub}_{\ell-1} (v_0''av_1'') = \fun{sub}_{\ell-1}(v'') = \fun{sub}_{\ell-1}(v_1'')$, which shows that $av_1'' \sim_{\ell-1} v_1''$. Similarly we can show that $bu_1'' \sim_{\ell-1} u_1''$. If $x=a$, then $w'u'' \sim_{\ell-1} u'' \sim_{\ell-1} v_1''$ implies that $aw'u'' \sim_{\ell-1} av_1'' \sim_{\ell-1} v_1'' \sim_{\ell-1} u''$. Similarly, if $x=b$, then $w'u'' \sim_{\ell-1} u'' \sim_{\ell-1} u_1''$ implies that $bw'u'' \sim_{\ell-1} bu_1'' \sim_{\ell-1} u_1'' \sim_{\ell-1} u''$. Therefore, $wu'' \sim_{\ell-1} u''$. Analogously, $wv'' \sim_{\ell-1} v''$.
    
    Finally, using the induction hypothesis (of the main statement) on $\A_p$, we get that $u''$ is accepted by $\A_p$ if and only if $wu''$ is accepted by $\A_p$, which is if and only if $u''$ is accepted by $\A_r$. Since $u'' \sim_{\ell-1} v''$, the induction hypothesis applied on $\A_r$ gives that $u''$ is accepted by $\A_r$ if and only if $v''$ is accepted by $\A_r$. However, this is if and only if $wv''$ is accepted by $\A_q$. Using the induction hypothesis on $\A_q$, we obtain that $wv''$ is accepted by $\A_q$ if and only if $v''$ is accepted by $\A_q$. Together, the assumption that $u''$ is accepted by $\A_p$ implies that $v''$ is accepted by $\A_q$. Hence $v=v'bv''$ is accepted by $\A$, which completes the proof.
  \end{proof}

  In other words, the previous theorem says that if $k$ is the minimum number for which a piecewise testable language $L$ is $k$-piecewise testable, then the depth of any ptNFA recognizing $L$ is at least $k$.
  
  It is natural to ask whether this property holds for any NFA recognizing the language $L$. The following result shows that it is not the case. Actually, for any natural number $\ell$, there exists a piecewise testable language such that the difference between its $k$-piecewise testability and the depth of an NFA is at least $\ell$.
  \begin{theorem}\label{propMain}
    For every $k\ge 3$, there exists a $k$-piecewise testable language that is recognized by an NFA of depth at most $\left\lfloor\frac{k}{2}\right\rfloor$.
  \end{theorem}
  \begin{proof}
    For every $i\ge 1$, let $L_i=a^i+a^{2i+1}\cdot a^*$. We show that the language $L_i$ is $(2i+1)$-piecewise testable and that there exists an NFA of depth at most $i$ recognizing it.
    
    The minimal DFA for $L_i$ consists of $2i+1$ states $\{0,1,\ldots,2i+1\}$, where $0$ is the initial state, $i$ and $2i+1$ are accepting, $p\cdot a = p+1$ for $p<2i+1$, and $(2i+1)\cdot a = 2i+1$. The depth is $2i+1$, which shows that $L_i$ is $(2i+1)$-piecewise testable. Notice that $a^{2i} \sim_{2i} a^{2i+1}$, but $a^{2i}$ does not belong to $L_i$, hence $L_i$ is not $2i$-piecewise testable.
    
    The NFA for $L_i$ consists of two cycles of length $i+1$, the structure is depicted in Figure~\ref{fig8}. The initial state is state 0 and the solely accepting state is state $i$. The automaton accepts $L_i$. Indeed, it accepts $a^i$ and no shorter word. After reading $a^i$, the automaton is in state $i$ or $i'$. In both cases, the shortest nonempty path to the single accepting state $i$ is of length $i+1$. Thus, the automaton accepts $a^{2i+1}$, but nothing between $a^i$ and $a^{2i+1}$. Finally, using the self-loop in state $i'$, the automaton accepts $a^ia^*a^{i+1} = a^{2i+1}a^*$. The depth of the automaton is $i$. 
    \begin{figure}
      \centering
      \begin{tikzpicture}[baseline,->,>=stealth,shorten >=1pt,node distance=1.6cm,
        state/.style={circle,minimum size=7mm,very thin,draw=black,initial text=},
        every node/.style={fill=white,font=\small}]
        \node[state,initial below]  (0) {$0$};
        \node[state]                (1') [above left of=0]  {$1'$};
        \node[state]                (2') [left of=1'] {$2'$};
        \node                       (3') [left of=2'] {$\ldots$};
        \node[state]                (i') [below left of=3'] {$i'$};
        \node[state]                (1) [above right of=0]  {$1$};
        \node[state]                (2) [right of=1]  {$2$};
        \node                       (3) [right of=2]  {$\ldots$};
        \node[state,accepting]      (i) [below right of=3]  {$i$};
        \path
          (0) edge[bend left] node {$a$} (1)
          (1) edge node {$a$} (2)
          (2) edge node {$a$} (3)
          (3) edge[bend left=50] node {$a$} (i)
          (i) edge node {$a$} (0)
          (0) edge[bend right] node {$a$} (1')
          (1') edge node {$a$} (2')
          (2') edge node {$a$} (3')
          (3') edge[bend right=50] node {$a$} (i')
          (i') edge node {$a$} (0)
          (i') edge[loop left] node {$a$} (i')
        ;
      \end{tikzpicture}
      \caption{The NFA of depth $i$ recognizing $L_i$}
      \label{fig8}
    \end{figure}
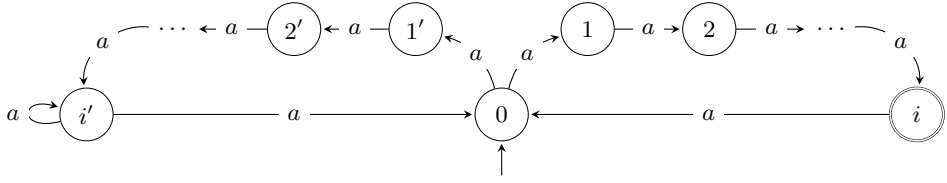
  \end{proof}

\subsection{Piecewise Testability and the Depth of NFAs}
  Theorem~\ref{thmMain} gives rise to a question, whether the opposite implication holds true.
  
  Notice that although the depth of ptNFAs is more suitable to provide bounds on $k$-piecewise testability, the depth is significantly influenced by the size of the input alphabet. For instance, for an alphabet $\Sigma$, the language $L = \bigcap_{a\in\Sigma} L_a$ of all words containing all letters of $\Sigma$ is a $1$-piecewise testable language such that any NFA recognizing it requires at least $2^{|\Sigma|}$ states and is of depth $|\Sigma|$, cf.~\cite{dlt15}.

  Considering the opposite direction of Theorem~\ref{thmMain}, it was independently shown in~\cite{KKP,dlt15} that, given a $k$-piecewise testable language over an $n$-letter alphabet, the tight upper bound on the depth of the minimal DFA recognizing it is $\binom{k+n}{k} - 1$. In other words, this formula gives the tight upper bound on the depth of the $\sim_k$-canonical DFA~\cite{dlt15} over an $n$ element alphabet. A related question on the size of this DFA is still open, see~\cite{Karandikar2015} for more details.

  We recall the result for DFAs.
  \begin{theorem}[\cite{KKP,dlt15}]\label{tightbound}
    For any natural numbers $k$ and $n$, the depth of the minimal DFA recognizing a $k$-piecewise testable language over an $n$-letter alphabet is at most $\binom{k+n}{k} - 1$. The bound is tight for any $k$ and $n$.
  \end{theorem}

  It remains open whether this is also a lower bound for NFAs or ptNFAs.

\section{Application and Discussion}\label{sec5}
  The reader might have noticed that the reverse of the automaton $\A_i$ constructed in Section~\ref{demo} is deterministic and, when made complete, it satisfies the conditions of Fact~\ref{thm:characterization}. Since, by definition, a language is $k$-piecewise testable if and only if its reverse is $k$-piecewise testable, this observation provides the same upper bound $i+1$ on $k$-piecewise testability of the language $L(\A_i)$. However, this is just a coincidence and it is not difficult to find an example of a ptNFA whose reverse is not deterministic.

  Since both the minimal DFA for $L$ and the minimal DFA for $L^R$ provide an upper bound on $k$, it could seem reasonable to compute both DFAs in parallel with the hope that (at least) one of them will be computed in a reasonable (polynomial) time. Although this may work for many cases (including the case of Section~\ref{demo}), we now show that there are cases where both the DFAs are of exponential size.

  \begin{theorem}\label{thmApp}
    For every $n\ge 0$, there exists a $(2n+1)$-state ptNFA $\B$ such that the depth of both the minimal DFA for $L(\B)$ and the minimal DFA for $L(\B)^R$ are exponential with respect to $n$.
  \end{theorem}
  \begin{proof}[Proof sketch.]
    The idea of the proof is to make use of the automaton $\A_i$ constructed in Section~\ref{demo} to build a ptNFA $\B_i$ such that $L(\B_i) = L(\A_i) \cdot L(\A_i)^R$. Then $L(\B_i) = L(\B_i)^R$ and it can be shown that the minimal DFA recognizing the language $L(\B_i)$ requires an exponential number of states compared to $\B_i$. Namely, the depth of both the minimal DFA for $L(\B_i)$ and the minimal DFA for $L(\B_i)^R$ are of length at least $2^{i+1}-1$.
  \end{proof}
 
  The previous proof provides another motivation to investigate nondeterministic automata for piecewise testable languages. Given several DFAs, the result of a sequence of operations may result in an NFA that preserves some good properties. Namely, the language $L(\B_i)$ from the previous proof is a result of the operation concatenation of a language $L^R$ with $L$, where $L$ is a piecewise testable language given as a DFA.
  
  It immediately follows from Theorem~\ref{thmMain} that the language $L(\B_i)$ is $(2i+1)$-piecewise testable. This result is not easily derivable from known results, which are either in PSPACE or require to compute an exponentially larger minimal DFA, which anyway provides only the information that the language $L(\B_i)$ is $k$-piecewise testable for some $k\ge 2^{i+1}-1$.
  
  Even the information that the language $L(\B_i) = L^R \cdot L$, for a piecewise testable language $L$, does not seem very helpful, since, as we show in the example below, piecewise testable languages are not closed under the concatenation even with its own reverse.
 
  \begin{example}\label{notClosedLLR}
    Let $L$ be the language over the alphabet $\{a,b,c\}$ defined by the regular expression $ab^* + c(a+b)^*$. The reader can construct the minimal DFA for $L$ and check that the properties of Fact~\ref{thm:characterization} are satisfied. In addition, the depth of the minimal DFA is two, hence the language is 2-piecewise testable. Since the properties of Theorem~\ref{thm1pt} (see below) are not satisfied, the language $L$ is not 1-piecewise testable. 
    
    On the other hand, the reader can notice that the sequence $ca,cab,caba,cabab,cababa,\ldots$ is an infinite sequence where every word on the odd position belongs to $L\cdot L^R$, whereas every word on the even position does not. This means that there exists a cycle in the minimal DFA recognizing $L\cdot L^R$, which shows that $L\cdot L^R$ is not a piecewise testable language according to Fact~\ref{thm:characterization}. The reader can also directly compute the minimal DFA for $L\cdot L^R$ and notice a non-trivial cycle in it. 
  \end{example}
 
  To complete this part, we show that the language $L(\B_i)$ is not $(2i)$-piecewise testable. Thus, there are no ptNFAs recognizing the language $L(\B_i)$ with depth less then $2i+1$.
  \begin{lemma}\label{lemma17}
    For every $i\ge 0$, the language $L(\B_i)$ is not $2i$-piecewise testable. 
  \end{lemma}

\section{Complexity}\label{sec4}
  In this section, we first give an overview of known complexity results and characterization theorems for DFAs and then discuss the related complexity for ptNFAs.
  
  Simon~\cite{Simon1972} proved that piecewise testable languages are exactly those regular languages whose syntactic monoid is $\mathcal{J}$-trivial, which shows decidability of the problem whether a regular language is piecewise testable. Later, Stern proved that the problem is decidable in polynomial time for languages represented as minimal DFAs~\cite{Stern85a}, and Cho and Huynh~\cite{ChoH91} showed that it is NL-complete for DFAs. Trahtman~\cite{Trahtman2001} improved Stern's result by giving an algorithm quadratic in the number of states of the minimal DFA, and Kl\'ima and Pol\'ak~\cite{KlimaP13} presented an algorithm quadratic in the size of the alphabet of the minimal DFA. If the language is represented as an NFA, the problem is PSPACE-complete~\cite{mfcs2014ex} (see more details below).

  By definition, a regular language is piecewise testable if there exists $k$ such that it is $k$-piecewise testable. It gives rise to a question to find such a minimal $k$. The $k$-piecewise testability problem asks, given an automaton, whether it recognizes a $k$-piecewise testable language. The problem is trivially decidable because there are only finitely many $k$-piecewise testable languages over a fixed alphabet. The coNP upper bound on $k$-piecewise testability for DFAs was independently shown in~\cite{HofmanM15,dlt15}.\footnote{Actually, \cite{HofmanM15} gives the bound NEXPTIME for the problem for NFAs where $k$ is part of the input. The coNP bound for DFAs can be derived from the proof omitted in the conference version. The problem is formulated in terms of separability, hence it requires the NFA for the language and for its complement.} The coNP-completeness for $k\ge 4$ was recently shown in~\cite{KKP}. The complexity holds even if $k$ is given as part of the input. The complexity analysis of the problem for $k<4$ is provided in~\cite{dlt15}. We recall the results we need later.
  
  \begin{theorem}[\cite{KKP}]\label{thmconp}
    For $k\ge 4$, to decide whether a DFA represents a $k$-piecewise testable language is coNP-complete.
    It remains coNP-complete even if the parameter $k\ge 4$ is given as part of the input. 
    For a fixed alphabet, the problem is decidable in polynomial time.
  \end{theorem}
  
  It is not difficult to see that, given a minimal DFA, it is decidable in constant time whether its language is 0-piecewise  testable, since it is either empty or $\Sigma^*$.
  \begin{theorem}[1-piecewise testability DFAs, \cite{dlt15}]\label{thm1pt}
    Let $\A=(Q,\Sigma,\cdot,i,F)$ be a minimal DFA. Then $L(\A)$ is 1-piecewise testable if and only if 
    (i) for every $p\in Q$ and $a\in\Sigma$, $paa=pa$ and
    (ii) for every $p\in Q$ and $a,b\in\Sigma$, $pab=pba$.
    The problem is in AC$^0$.
  \end{theorem}
  
  It is not hard to see that this result does not hold for ptNFAs. Indeed, one can simply consider a minimal DFA satisfying the properties and add a nondeterministic transition that violates them, but not the properties of ptNFAs. On the other hand, the conditions are still sufficient.
  \begin{lemma}[1-piecewise testability ptNFAs]\label{thm1ptNFAs}
    Let $\A=(Q,\Sigma,\cdot,i,F)$ be a complete NFA. If 
    (i) for every $p\in Q$ and $a\in\Sigma$, $paa=pa$ and
    (ii) for every $p\in Q$ and $a,b\in\Sigma$, $pab=pba$,
    then the language $L(\A)$ is 1-piecewise testable.
  \end{lemma}
  
  Note that any ptNFA $\A$ satisfying $(i)$ must have $|pa|=1$ for every state $p$ and letter $a$. If $pa=\{r_1,r_2,\ldots,r_m\}$ with $r_1 < r_2 < \ldots < r_m$, then $paa=pa$ implies that $\{r_1,\ldots,r_m\} a = \{r_1,\ldots,r_m\}$. Then $r_1 \in r_1 a$ and the UMS property says that $r_1 a = \{r_1\}$. By induction, we can show hat $r_i a = \{r_i\}$. Consider the component of $G(\A,\Sigma(r_1))$ containing $r_1$. Then $r_1,\ldots,r_m$ all belong to this component. Since $r_1$ is maximal, $r_1$ is reachable from every $r_i$ under $\Sigma(r_1)\supseteq \{a\}$. However, the partial order $r_1 < \ldots < r_m$ implies that $r_1$ is reachable from $r_i$ only if $r_i=r_1$. Thus, $|pa|=1$. However, $\A$ can still have many initial states, which can be seen as a finite union of piecewise testable languages rather then a nondeterminism.

  The 2-piecewise testability characterization for DFAs is as follows.
  \begin{theorem}[2-piecewise testability DFAs, \cite{dlt15}]\label{thm2ptNL}
    Let $\A=(Q,\Sigma,\cdot,i,F)$ be a minimal partially ordered and confluent DFA. The language $L(\A)$ is 2-piecewise testable if and only if for every $a\in\Sigma$ and every state $s$ such that $i w = s$ for some $w\in\Sigma^*$ with $|w|_a\ge 1$, $s ba = s aba$ for every $b\in\Sigma\cup\{\eps\}$.
    The problem is NL-complete.
  \end{theorem}
  
  It is again sufficient for ptNFAs.
  \begin{lemma}[2-piecewise testability ptNFAs]\label{thm2ptNFAs}
    Let $\A=(Q,\Sigma,\cdot,i,F)$ be a ptNFA. If for every $a\in\Sigma$ and every state $s$ such that $i w = s$ for some $w\in\Sigma^*$ with $|w|_a\ge 1$, $s ba = s aba$ for every $b\in\Sigma\cup\{\eps\}$, then the language $L(\A)$ is 2-piecewise testable.
  \end{lemma}

  Considering Theorem~\ref{thmconp}, the lower bound for DFAs is indeed a lower bound for ptNFAs. Thus, we immediately have that the $k$-piecewise testability problem for ptNFAs is coNP-hard for $k\ge 4$. We now show that it is actually coNP-hard for every $k\ge 0$. The proof is split into two lemmas.

  The proof of the following lemma is based on the proof that the non-equivalence problem for regular expressions with operations union and concatenation is NP-complete, even if one of them is of the form $\Sigma^n$ for some fixed $n$~\cite{Hunt73,StockmeyerM73}.

  \begin{lemma}\label{0ptNFAhard}
    The 0-piecewise testability problem for ptNFAs is coNP-hard (even if the alphabet is binary).
  \end{lemma}

  It seems natural that the $(k+1)$-piecewise testability problem is not easier then the $k$-piecewise testability problem. We now formalize this intuition. We also point out that our reduction introduces a new symbol to the alphabet.
  \begin{lemma}\label{lemKtoK+1}
    For $k\ge 0$, $k$-piecewise testability is polynomially reducible to $(k+1)$-piecewise testability.
  \end{lemma}

  Together, since the $k$-piecewise testability problem for NFAs is in PSPACE~\cite{dlt15}, we have the following result.
  \begin{theorem}\label{thmArbCoNP}
    For $k\ge 0$, the $k$-piecewise testability problem for ptNFAs is coNP-hard and in PSPACE.
  \end{theorem}

  \subparagraph*{The case of a fixed alphabet.}
  The previous discussion is for the general case where the alphabet is arbitrary and considered as part of the input. In this subsection, we assume that the alphabet is fixed. In this case, it is shown in the arxiv versions v1--v4 of~\cite{Karandikar2015} that the length of the shortest representatives of the $\sim_k$-classes is bounded by the number $\left( \frac{k+2c-1}{c} \right)^c$, where $c$ is the cardinality of the alphabet. This gives us the following result for 0-piecewise testability for ptNFAs.
  
  \begin{lemma}\label{0-PTcoNPc}
    For a fixed alphabet $\Sigma$ with $c=|\Sigma|\ge 2$, the 0-piecewise testability problem for ptNFAs is coNP-complete. 
  \end{lemma}
  \begin{proof}
    The hardness follows from Lemma~\ref{0ptNFAhard}, since it is sufficient to use a binary alphabet.
    
    We now prove completeness. Let $\A$ be a ptNFA over $\Sigma$ of depth $d$ recognizing a nonempty language (this can be checked in NL). Then the language $L(\A)$ is $d$-piecewise testable by Theorem~\ref{thmMain}. This means that if $v\sim_d u$, then either both $u$ and $v$ are accepted or both are rejected by $\A$. Now, the language $L(\A)\neq\emptyset$ is not 0-piecewise testable if and only if $L(\A)$ is non-universal. Since $\Sigma$ is fixed, the shortest representative of any of the $\sim_d$-classes is of length less than $\left( \frac{d+2c-1}{c} \right)^c = O(d^c)$, which is polynomial in the depth of $\A$. Thus, if the language $L(\A)$ is not universal, then the nondeterministic algorithm can guess a shortest representative of a non-accepted $\sim_d$-class and verify the guess in polynomial time.
  \end{proof}

  We can now generalize this result to $k$-piecewise testability.
  \begin{theorem}\label{theorem27}
    Let $\Sigma$ be a fixed alphabet with $c=|\Sigma|\ge 2$, and let $k\ge 0$. Then the problem to decide whether the language of a ptNFA $\A$ over $\Sigma$ is $k$-piecewise testable is coNP-complete.
  \end{theorem}

  Note that this is in contrast with the analogous result for DFAs, cf. Theorem~\ref{thmconp}, where the problem is in P for DFAs over a fixed alphabet. In addition, the hardness part of the previous proof gives us the following corollary, which does not follow from the hardness proof of~\cite{KKP}, since the proof there requires a growing alphabet.
  \begin{corollary}\label{cor28}
    The $k$-piecewise testability problem for ptNFAs over an alphabet $\Sigma$ is coNP-hard for $k\ge 0$ even if $|\Sigma|=3$.
  \end{corollary}

  \subparagraph*{The case of a unary alphabet.}
  Since Lemma~\ref{0-PTcoNPc} (resp. Lemma~\ref{0ptNFAhard}) requires at least two letters in the alphabet to prove coNP-hardness, it remains to consider the case of a unary alphabet. We now show that the problem is simpler, unless P=coNP.
  Namely, a similar argument as in the proof of Lemma~\ref{0-PTcoNPc}, improved by the fact that the length of the shortest representatives of $\sim_k$-classes is bounded by the depth of the ptNFA, gives the following result. 
  \begin{theorem}\label{thmP}
    The $k$-piecewise testability problem for ptNFAs over a unary alphabet is decidable in polynomial time. The result holds even if $k$ is given as part of the input.
  \end{theorem}

  In contrast to this, we now show that the problem is coNP-complete for general NFAs.
  \begin{theorem}\label{thm30}
    Both piecewise testability and $k$-piecewise testability problems for NFAs over a unary alphabet are coNP-complete.
  \end{theorem}
  
  The complexity of $k$-piecewise testability for considered automata is summarized in Table~\ref{table1}. Note that the precise complexity of $k$-piecewise testability for ptNFAs is not yet known in the case the alphabet is consider as part of the input even for $k=0$.
  \begin{table}
    \begin{center}
      \begin{tabular}{r|c|c|c|c|}
        & Unary alphabet & Fixed alphabet & \multicolumn{2}{c|}{Arbitrary alphabet}\\
        &                &                & $k\le 3$ & $k\ge 4$\\
        \hline
        DFA   & P
              & P \cite{KKP}
              & NL-complete \cite{dlt15}
              & coNP-complete \cite{KKP}\\
        ptNFA & P 
              & coNP-complete  
              & \multicolumn{2}{c|}{PSPACE \& coNP-hard}\\ 
          NFA & coNP-complete
              & PSPACE-complete \cite{dlt15}
              & \multicolumn{2}{c|}{PSPACE-complete \cite{dlt15}}
      \end{tabular}
    \end{center}
    \caption{Complexity of $k$-piecewise testability -- an overview}
    \label{table1}
  \end{table}

\section{Conclusion}\label{conclusion}
  In this paper, we have defined a class of nondeterministic finite automata (ptNFAs) that characterize piecewise testable languages. We have shown that their depth (exponentially) improves the known upper bound on $k$-piecewise testability shown in~\cite{KlimaP13} for DFAs. We have discussed several related questions, mainly in comparison with DFAs and NFAs, including the complexity of $k$-piecewise testability for ptNFAs. It can be noticed that the results for ptNFAs generalize the results for DFAs in the sense that the results for DFAs are consequences of the results presented here. This, however, does not hold for the complexity results.

  \subparagraph*{The length of a shortest proof over an arbitrarily alphabet.}
  It is an open question what is the complexity of $k$-piecewise testability if the alphabet is consider as part of the input. Notice that the results of~\cite{Karandikar2015} give a lower bound on the maximal length of the shortest representative of a class. Namely, let $L_k(n)$ denote the maximal length of the shortest representatives of the $\sim_k$-classes over an $n$-element alphabet. Then $(L_n(k) +1) \log n > (\frac{k}{n})^{n-1} \log (\frac{k}{n})$. Setting $k=n^2$ then gives that $L_n(n^2) > n^{n-1}$. Thus, the representative can be of exponential length with respect to the size of the alphabet. However, how many states does a ptNFA require to exclude such a representative while accepting every shorter word?

\subparagraph*{Acknowledgements.}
  We thank the authors of~\cite{HofmanM15} and~\cite{KKP} for providing us with full versions of their papers.

\bibliographystyle{plain}
\bibliography{mybib}

\begin{thebibliography}{10}

\bibitem{AhoHU74}
A.~V. Aho, J.~E. Hopcroft, and J.~D. Ullman.
\newblock {\em The Design and Analysis of Computer Algorithms}.
\newblock Addison-Wesley, 1974.

\bibitem{Almeida2008486}
J.~Almeida, J.~C. Costa, and M.~Zeitoun.
\newblock Pointlike sets with respect to {R} and {J}.
\newblock {\em Journal of Pure and Applied Algebra}, 212(3):486--499, 2008.

\bibitem{AlmeidaZ-ita97}
J.~Almeida and M.~Zeitoun.
\newblock The pseudovariety {J} is hyperdecidable.
\newblock {\em RAIRO -- Theoretical Informatics and Applications},
  31(5):457--482, 1997.

\bibitem{Bojanczyk:2012}
M.~Bojanczyk, L.~Segoufin, and H.~Straubing.
\newblock Piecewise testable tree languages.
\newblock {\em Logical Methods in Computer Science}, 8(3), 2012.

\bibitem{ChoH91}
S.~Cho and D.~T. Huynh.
\newblock Finite-automaton aperiodicity is {PSPACE}-complete.
\newblock {\em Theoretical Computer Science}, 88(1):99--116, 1991.

\bibitem{CohenB71}
R.~S. Cohen and J.~A. Brzozowski.
\newblock Dot-depth of star-free events.
\newblock {\em Journal of Computer and System Sciences}, 5(1):1--16, 1971.

\bibitem{icalp2013}
W.~Czerwi{\'n}ski, W.~Martens, and T.~Masopust.
\newblock Efficient separability of regular languages by subsequences and
  suffixes.
\newblock In {\em ICALP}, volume 7966 of {\em LNCS}, pages 150--161, 2013.

\bibitem{CzerwinskiM14}
W.~Czerwi{\'n}ski, W.~Martens, L.~van Rooijen, and M.~Zeitoun.
\newblock A note on decidable separability by piecewise testable languages.
\newblock In {\em {FCT}}, volume 9210 of {\em LNCS}, pages 173--185, 2015.

\bibitem{DiekertGK08}
V.~Diekert, P.~Gastin, and M.~Kufleitner.
\newblock A survey on small fragments of first-order logic over finite words.
\newblock {\em Int. Journal of Foundations of Computer Science},
  19(3):513--548, 2008.

\bibitem{FuHT2011}
J.~Fu, J.~Heinz, and H.~G. Tanner.
\newblock An algebraic characterization of strictly piecewise languages.
\newblock In {\em TAMC}, volume 6648 of {\em LNCS}, pages 252--263. 2011.

\bibitem{GarciaR04}
P.~Garc{\'{\i}}a and J.~Ruiz.
\newblock Learning $k$-testable and $k$-piecewise testable languages from
  positive data.
\newblock {\em Grammars}, 7:125--140, 2004.

\bibitem{GarciaV90}
P.~Garc{\'{\i}}a and E.~Vidal.
\newblock Inference of k-testable languages in the strict sense and application
  to syntactic pattern recognition.
\newblock {\em {IEEE} Transactions on Pattern Analysis and Machine
  Intelligence}, 12(9):920--925, 1990.

\bibitem{HofmanM15}
P.~Hofman and W.~Martens.
\newblock Separability by short subsequences and subwords.
\newblock In {\em {ICDT}}, volume~31 of {\em LIPIcs}, pages 230--246, 2015.

\bibitem{mfcs2014}
{\v S}.~Holub, G.~Jir\'askov\'a, and T.~Masopust.
\newblock On upper and lower bounds on the length of alternating towers.
\newblock In {\em {MFCS}}, volume 8634 of {\em LNCS}, pages 315--326, 2014.

\bibitem{mfcs2014ex}
{\v S}.~Holub, T.~Masopust, and M.~Thomazo.
\newblock Alternating towers and piecewise testable separators.
\newblock {\em CoRR}, abs/1409.3943, 2014.

\bibitem{Hunt73}
H.~B. {Hunt III}.
\newblock {\em On the Time and Tape Complexity of Languages}.
\newblock PhD thesis, Department of Computer Science, Cornell University,
  Ithaca, NY, 1973.

\bibitem{Karandikar2015}
P.~Karandikar, M.~Kufleitner, and Ph. Schnoebelen.
\newblock On the index of {S}imon's congruence for piecewise testability.
\newblock {\em Information Processing Letters}, 115(4):515--519, 2015.

\bibitem{KKP}
O.~Kl{\'{\i}}ma, M.~Kunc, and L.~Pol{\'{a}}k.
\newblock Deciding $k$-piecewise testability.
\newblock Submitted.

\bibitem{KlimaP13}
O.~Kl\'{\i}ma and L.~Pol{\'a}k.
\newblock Alternative automata characterization of piecewise testable
  languages.
\newblock In {\em {DLT}}, volume 7907 of {\em LNCS}, pages 289--300, 2013.

\bibitem{Kontorovich2008}
L.~Kontorovich, C.~Cortes, and M.~Mohri.
\newblock Kernel methods for learning languages.
\newblock {\em Theoretical Computer Science}, 405(3):223--236, 2008.

\bibitem{KufleitnerL12}
M.~Kufleitner and A.~Lauser.
\newblock Around dot-depth one.
\newblock {\em International Journal of Foundations of Computer Science},
  23(6):1323--1340, 2012.

\bibitem{dlt15}
T.~Masopust and M.~Thomazo.
\newblock On the complexity of $k$-piecewise testability and the depth of
  automata.
\newblock In {\em {DLT}}, volume 9168 of {\em LNCS}, pages 364--376, 2015.

\bibitem{Myhill}
J.~Myhill.
\newblock Finite automata and representation of events.
\newblock Technical report, Wright Air Development Center, 1957.

\bibitem{PerrinP86}
D.~Perrin and J.{-}E. Pin.
\newblock First-order logic and star-free sets.
\newblock {\em Journal of Computer and System Sciences}, 32(3):393--406, 1986.

\bibitem{PerrinPin}
D.~Perrin and J.-E. Pin.
\newblock {\em Infinite words: Automata, semigroups, logic and games}, volume
  141 of {\em Pure and Applied Mathematics}.
\newblock 2004.

\bibitem{mfcsPlaceRZ13}
T.~Place, L.~van Rooijen, and M.~Zeitoun.
\newblock Separating regular languages by piecewise testable and unambiguous
  languages.
\newblock In {\em MFCS}, volume 8087 of {\em LNCS}, pages 729--740, 2013.

\bibitem{Rogers:2007}
J.~Rogers, J.~Heinz, G.~Bailey, M.~Edlefsen, M.~Visscher, D.~Wellcome, and
  S.~Wibel.
\newblock On languages piecewise testable in the strict sense.
\newblock In {\em MOL}, volume 6149 of {\em LNAI}, pages 255--265, 2010.

\bibitem{RogersHFHLW12}
J.~Rogers, J.~Heinz, M.~Fero, J.~Hurst, D.~Lambert, and S.~Wibel.
\newblock Cognitive and sub-regular complexity.
\newblock In {\em FG}, volume 8036 of {\em LNCS}, pages 90--108, 2013.

\bibitem{Simon1972}
I.~Simon.
\newblock {\em Hierarchies of Events with Dot-Depth One}.
\newblock PhD thesis, Department of Applied Analysis and Computer Science,
  University of Waterloo, Canada, 1972.

\bibitem{Stern85a}
J.~Stern.
\newblock Complexity of some problems from the theory of automata.
\newblock {\em Information and Control}, 66(3):163--176, 1985.

\bibitem{StockmeyerM73}
L.~J. Stockmeyer and A.~R. Meyer.
\newblock Word problems requiring exponential time: Preliminary report.
\newblock In {\em {STOC}}, pages 1--9. {ACM}, 1973.

\bibitem{Thomas82}
W.~Thomas.
\newblock Classifying regular events in symbolic logic.
\newblock {\em Journal of Computer and System Sciences}, 25(3):360--376, 1982.

\bibitem{Trahtman2001}
A.~N. Trahtman.
\newblock Piecewise and local threshold testability of {DFA}.
\newblock In {\em {FCT}}, volume 2138 of {\em LNCS}, pages 347--358, 2001.

\bibitem{lvanrooijen}
L.~van Rooijen.
\newblock {\em A combinatorial approach to the separation problem for regular
  languages}.
\newblock PhD thesis, LaBRI, University of Bordeaux, France, 2014.

\end{thebibliography}

\section{Proofs of Section~\ref{secDepth}}

  \begin{replemma}{lemmaAi}
    Let $\A$ be a ptNFA with $I$ denoting the set of initial states. Then the language $L(\A)=\bigcup_{i\in I} L(\A_i)$, where every sub-automaton $\A_i$ is a ptNFA.
  \end{replemma}
  \begin{proof}
    Indeed, $L(\A)=\bigcup_{i\in I} L(\A_i)$ holds. It remains to show that every $\A_i$ is partially order, complete, and satisfies the UMS property. However, $\A_i$ is obtained from $\A$ by removing the states not reachable from $i$ and the corresponding transitions. Since $\A$ is complete and partially ordered, so is $\A_i$. If the UMS property was not satisfied in $\A_i$, it would not be satisfied in $\A$ either, hence $\A_i$ satisfies the UMS property.
  \end{proof}

\section{Proofs of Section~\ref{sec5}}

  \begin{reptheorem}{thmApp}
    For every $n\ge 0$, there exists a $(2n+1)$-state ptNFA $\B$ such that the depth of both the minimal DFA for $L(\B)$ and the minimal DFA for $L(\B)^R$ are exponential with respect to $n$.
  \end{reptheorem}
  \begin{proof}
    The idea of the proof is to make use of the automaton $\A_i$ constructed in Section~\ref{demo} to build a ptNFA $\B_i$ such that $L(\B_i) = L(\A_i) \cdot L(\A_i)^R$. Then $L(\B_i) = L(\B_i)^R$ and we show that the minimal DFA recognizing the language $L(\B_i)$ requires an exponential number of states compared to $\B_i$.
    Thus, for every $i\ge 0$, we define the NFA
    \[
      \B_{i}=(\{-i,\ldots,-1,0,1,\ldots,i\},\{a_0,a_1,\ldots,a_i\},\cdot,I_i,-I_i)
    \]
    with $I_i = \{0,1,\ldots,i\}$ and the transition function $\cdot$ defined so that $j \cdot a_\ell = j$ if $i \ge |j| > \ell \ge 0$, $\ell \cdot a_\ell = \{0,1,\ldots,\ell-1\}$, and $-j \cdot a_\ell = -\ell$ if $0\le j < \ell \le i$. Automaton $\B_2$ is depicted in Figure~\ref{fig7}. 
    
    Notice that $L(\B_{i-1})\subseteq L(\B_{i})$ and that $\B_i$ has $2i+1$ states. The reader can see that $L(\B_i) = L(\B_i)^R$. Moreover, making the NFA $\B_i$ complete (the dotted lines in Figure~\ref{fig7}), results in a ptNFA. Therefore, the language $L(\B_i)$ is piecewise testable by Theorem~\ref{thm10}.

    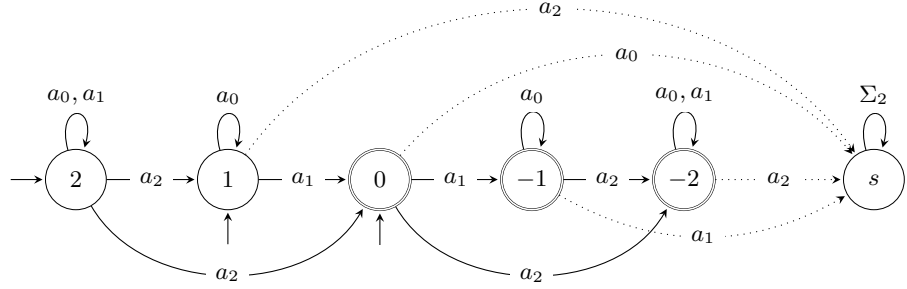
\begin{figure}
      \centering
      \begin{tikzpicture}[baseline,->,>=stealth,shorten >=1pt,node distance=2cm,
        state/.style={circle,minimum size=8mm,very thin,draw=black,initial text=},
        every node/.style={fill=white,font=\small},
        bigloop/.style={shift={(0,0.01)},text width=.8cm,align=center}]
        \node[state,initial below,accepting]    (1) {$0$};
        \node[state,initial below]              (4) [left of=1]     {$1$};
        \node[state,initial]                    (5) [left of=4]     {$2$};
        \node[state,accepting]                  (4') [right of=1]   {$-1$};
        \node[state,accepting]                  (5') [right of=4']  {$-2$};
        \node[state]                            (s) [right of=5',node distance=2.5cm]  {$s$};
        \path
          (4) edge node {$a_1$} (1)
          (4) edge[loop above] node[bigloop] {$a_0$} (4)
          (5) edge[loop above] node[bigloop] {$a_0,a_1$} (5)
          (5) edge[bend right=60] node {$a_2$} (1)
          (5) edge node {$a_2$} (4)
          (1) edge node {$a_1$} (4')
          (4') edge[loop above] node[bigloop] {$a_0$} (4')
          (5') edge[loop above] node[bigloop] {$a_0,a_1$} (5')
          (1) edge[bend right=60] node {$a_2$} (5')
          (4') edge node {$a_2$} (5')
          (4) edge[dotted,bend left=55] node {$a_2$} (s)
          (1) edge[dotted,bend left=50] node {$a_0$} (s)
          (4') edge[dotted,bend right] node {$a_1$} (s)
          (5') edge[dotted] node {$a_2$} (s)
          (s) edge[loop above] node[bigloop] {$\Sigma_2$} (s)
          ;
      \end{tikzpicture}
      \caption{Automaton $\B_2$ (without dotted transitions) and its completion (with dotted transitions)}
      \label{fig7}
    \end{figure}

    We now define a word $w_i$ inductively by $w_0=a_0$ and $w_{\ell} = w_{\ell-1} a_{\ell} w_{\ell-1}$, for $0 < \ell \le i$. Then $|w_i| = 2^{i+1}-1$ and we show that every prefix of $w_i$ of even length belongs to $L(\B_i)$ and every prefix of odd length does not. 
    
    Indeed, $\eps$ belongs to $L(\B_0)\subseteq L(\B_i)$. Let $v$ be a prefix of $w_i$ of even length. If $|v| < 2^i-1$, then $v$ is a prefix of $w_{i-1}$ and $v\in L(\B_{i-1})\subseteq L(\B_i)$ by the induction hypothesis. If $|v| > 2^i-1$, then $v = w_{i-1} a_i v'$, where $v'$ is a prefix of $w_{i-1}$ of even length. The definition of $\B_i$ and the induction hypothesis then imply that there is a path $i\xrightarrow{w_{i-1}} i\xrightarrow{a_i} (i-1) \xrightarrow{v'} 0$. Thus, $v$ belongs to $L(\B_i)$.
    
    We now show that any prefix $w$ of $w_i$ of odd length does not belong to $L(\B_i)$. Since $w$ begins and ends with $a_0$ and there is neither an $a_0$-transition to nor from state 0, it cannot be accepted either by or from state $0$. Therefore, if $w$ is accepted by $\B_i$, there must be an accepting computation starting from an initial state $q_0 \in \{1,\ldots,i\}$ and ending in an accepting state $q_f \in \{-1,\ldots,-i\}$. It means that $w$ can be written as $w = u a_\ell a_j v$, where $q_0 \xrightarrow{ua_\ell} 0 \xrightarrow{a_j v} q_f$. By the construction, both $\ell$ and $j$ are different from $0$, which is a contradiction with the structure of $w_{i}$, since $a_0$ is on every odd position.
    
    These properties imply that the prefixes of $w_i$ alternate between accepting and non-accepting states of the minimal DFA for $L(\B_i)$. Since the language $L(\B_i)$ is piecewise testable, the minimal DFA does not have any non-trivial cycles. Thus, the word $w_i$ forms a simple path in the minimal DFA recognizing the language $L(\B_i)$, which shows that the depth of the minimal DFA is of length at least $2^{i+1}-1$.
  \end{proof}

  \begin{replemma}{lemma17}
    For every $i\ge 0$, the language $L(\B_i)$ is not $2i$-piecewise testable. 
  \end{replemma}
  \begin{proof}
    Let $w_i = w_{i-1} a_i w_{i-1}$ be the word as defined in the proof of Theorem~\ref{thmMain}, and let $w_i'$ denote its prefix without the last letter, that is, $w_i = w_i' a_0$. 
    We show that $w_i'a_0(w_i')^R \sim_{2i} w_i'(w_i')^R$. Combining this with the observation that $w_i'a_0(w_i')^R$ does not belong to $L(\B_i)$ and $w_i'(w_i')^R$ belongs to $L(\B_i)$ then implies that $L(\B_i)$ is not $2i$-piecewise testable.

    Indeed, $w_i'(w_i')^R \preccurlyeq w_i'a_0(w_i')^R$, therefore we need to show that if $w \in \fun{sub}_{2i}(w_i'a_0(w_i')^R)$, then $w\in\fun{sub}_{2i}(w_i'(w_i')^R)$. If $w$ can be embedded into $w_i'a_0(w_i')^R$ without mapping $a_0$ of $w$ to the $a_0$ between $w_i'$ and $(w_i')^R$, then the claim holds. Thus, assume that $w = u a_0 v$ is such that the $a_0$ must be mapped to the $a_0$ between $w_i'$ and $(w_i')^R$. Thus, $u$ must be embedded into $w_i'$. We show by induction on $i$ that the length of $u$ must be at least $i$. It obviously holds for $i=0$. Assume that the claim holds for $i-1$ and consider $w_{i}' a_0 = w_{i-1} a_{i} w_{i-1}' a_0$. Since the $a_0$ of $w$ must be mapped to the last letter of $w_{i}' a_0$ and $\alp(w_{i-1}a_{i}) = \{a_0,a_1,\ldots,a_{i}\}$, there must be a nonempty prefix $u_1$ of $u$, i.e., $u=u_1u'$, such that $u_1$ is embedded into $w_{i-1} a_{i}$ and it forces the first letter of $u'$ to be embedded to $w_{i-1}' a_0$ in $w_{i}' a_0$. We now have that $u' a_0$ is embedded into $w_{i-1}' a_0$ such that $a_0$ must be mapped to the last letter of $w_{i-1}' a_0$. By the induction hypothesis, the length of $u'$ is at least $i-1$. Since $u_1$ is nonempty, we obtain that the length of $u=u_1u'$ is at least $i$. 
    
    Since the word $w_i' a_0$ is a palindrome, the same argument applies to $v$. Together, we have that $|w| = |u|+1+|v| \ge 2i+1$, which is a contradiction with the assumption that $|w|\le 2i$.
  \end{proof}

\section{Proofs of Section~\ref{sec4}}

  \begin{replemma}{thm1ptNFAs}[1-piecewise testability ptNFAs]
    Let $\A=(Q,\Sigma,\cdot,i,F)$ be a complete NFA. If 
    (i) for every $p\in Q$ and $a\in\Sigma$, $paa=pa$ and
    (ii) for every $p\in Q$ and $a,b\in\Sigma$, $pab=pba$,
    then the language $L(\A)$ is 1-piecewise testable.
  \end{replemma}
  \begin{proof}
    Consider the minimal DFA $\D$ constructed from $\A$ by the standard subset construction and minimization. We show that $\D$ satisfies the properties of Theorem~\ref{thm1pt}, which then implies the claim. Because every state of $\D$ is represented by a nonempty subset of states of $\A$, let $X\subseteq Q$ be a state of $\D$. Then, we have that $Xaa = \bigcup_{p\in X} paa = \bigcup_{p\in X} pa = Xa$ and, similarly, that $Xab = \bigcup_{p\in X} pab = \bigcup_{p\in X} pba = Xba$. Theorem~\ref{thm1pt} then completes the proof.
  \end{proof}

  \begin{replemma}{thm2ptNFAs}[2-piecewise testability ptNFAs]
    Let $\A=(Q,\Sigma,\cdot,i,F)$ be a ptNFA. If for every $a\in\Sigma$ and every state $s$ such that $i w = s$ for some $w\in\Sigma^*$ with $|w|_a\ge 1$, $s ba = s aba$ for every $b\in\Sigma\cup\{\eps\}$, then the language $L(\A)$ is 2-piecewise testable.
  \end{replemma}
  \begin{proof}
    Consider the minimal DFA $\D$ obtain from $\A$ by the standard subset construction and minimization. Since any ptNFA recognizes a piecewise testable language, see~Theorem~\ref{thm10}, $\D$ is confluent and partially ordered. We now show that it satisfies the properties of Theorem~\ref{thm2ptNL}. 
    
    Again, the states of $\D$ are represented by nonempty subsets of $\A$. Let $I\subseteq Q$ denote the initial state of $\D$. Let $a\in\Sigma$, and let $w\in \Sigma^*$ be such that $|w|_a \ge 1$. Denote $Iw = S$ and consider any $b\in\Sigma\cup\{\eps\}$. Then, since $sba = saba$ in $\A$, $Sba = \bigcup_{s\in S} sba = \bigcup_{s\in S} saba = Saba$.
  \end{proof}

  \begin{replemma}{0ptNFAhard}
    The 0-piecewise testability problem for ptNFAs is coNP-hard.
  \end{replemma}
  \begin{proof}
    We reduce the complement of CNF satisfiability. Let $U=\{x_1,x_2,\ldots,x_n\}$ be a set of variables and $\varphi = \varphi_1 \land \varphi_2 \land \ldots \land \varphi_m$ be a formula in CNF, where every $\varphi_i$ is a disjunction of literals. Without loss of generality, we may assume that no clause $\varphi_i$ contains both $x$ and $\neg x$. Let $\neg \varphi$ be the negation of $\varphi$ obtained by the de Morgan's laws. Then $\neg\varphi = \neg\varphi_1 \lor \neg\varphi_2 \lor \ldots \lor \neg\varphi_m$ is in DNF. For every $i=1,\ldots,m$, define $\beta_i = \beta_{i,1}\beta_{i,2}\ldots\beta_{i,n}$, where 
    \[
      \beta_{i,j} = \left\{
        \begin{array}{ll}
          0+1 & \text{ if } x_j \text{ and } \neg x_j \text{ do not appear in } \neg\varphi_i\\
          0   & \text{ if } \neg x_j \text{ appears in } \neg\varphi_i\\
          1   & \text{ if } x_j \text{ appears in } \neg\varphi_i
        \end{array}
        \right.
    \]
    for $j=1,2,\ldots,n$. Let $\beta = \bigcup_{i=1}^{m} \beta_{i}$. Then $w\in L(\beta)$ if and only if $w$ satisfies some $\neg\varphi_i$. That is, $L(\beta) = \{0,1\}^n$ if and only if $\neg\varphi$ is a tautology, which is if and only if $\varphi$ is not satisfiable. Note that by the assumption, the length of every $\beta_{i}$ is exactly $n$.

    We construct a ptNFA $\M$ as follows (the transitions are the minimal sets satisfying the definitions). The initial state of $\M$ is state $0$. For every $\beta_{i}$, we construct a deterministic path consisting of $n+1$ states $\{q_{i,0},q_{i,1},\ldots,q_{i,n}\}$ with transitions $q_{i,\ell+1} \in q_{i,\ell} \cdot \beta_{i,\ell}$ and $q_{i,0} = 0$. In addition, we add $n+1$ states $\{\alpha_1,\alpha_2,\ldots,\alpha_{n+1}\}$ and transitions $\alpha_{\ell+1} \in \alpha_\ell \cdot a$, for $\ell < n+1$ and $\alpha_0 = 0$, and $\alpha_{n+1} \in \alpha_{n+1} \cdot a$, where $a\in \{0,1\}$. This path is used to accept all words of length different from $n$. Finally, we add $n$ states $\{r_1,\ldots,r_n\}$ and transitions $r_{i+1} \in r_i \cdot a$, for $i<n$, and $\alpha_{n+1} \in r_n \cdot a$, where $a\in \{0,1\}$. These states are used to complete $\M$ by adding a transition from every state $q$ to $r_1$ under $a$ if $a$ is not defined in $q$. They ensure that any word of length $n$ that does not belong to $L(\beta)$ is not accepted by $\M$. The accepting states of $\M$ are the states $\{0,q_{1,n},\ldots,q_{m,n}\}\cup \{\alpha_1,\ldots\alpha_{n+1}\}\setminus\{\alpha_n\}$. Notice that $\M$ is partially ordered, complete and satisfies the UMS property. Indeed, the UMS property is satisfied since the only state with self-loops is the unique maximal state $\alpha_{n+1}$. The automaton accepts the language $L(\M) = L(\beta) \cup \{w \in \{0,1\}^* \mid |w| \neq n\}$.
    
    By Theorem~\ref{thm10}, the language is piecewise testable. It is 0-piecewise testable if and only if $L(\M)=\{0,1\}^*$, which is if and only if $L(\beta) = \{0,1\}^n$.
  \end{proof}

  \begin{replemma}{lemKtoK+1}
    For $k\ge 0$, $k$-piecewise testability is polynomially reducible to $(k+1)$-piecewise testability.
  \end{replemma}
  \begin{proof}
    Let $L_k$ over $\Sigma_k$ be a piecewise testable language recognized by a ptNFA $\M_k$ with the set of initial states $I_k=\{i_1,\ldots,i_\ell\}$. We construct the language $L_{k+1}$ over the alphabet $\Sigma_{k+1} = \Sigma_k \cup \{a_{k+1}\}$, where $a_{k+1} \notin \Sigma_k$, as depicted in Figure~\ref{reduction}. Namely, $\M_{k+1}$ recognizing the language $L_{k+1}$ is constructed from $\M_k$ by adding self-loops under $a_{k+1}$ to every state of $\M_k$ and adding, for every initial state $i$ of $\M_k$, a new state $i'$ that contains self-loops under all letters from $\Sigma_k$ and goes to the initial states $i$ of $\M_k$ under $a_{k+1}$. The initial states of $\M_{k+1}$ are the new states $i'$, the accepting states are the accepting states of $\M_k$. Notice that the automaton $\M_{k+1}$ is a ptNFA.
    \begin{figure}
      \centering
      \includegraphics{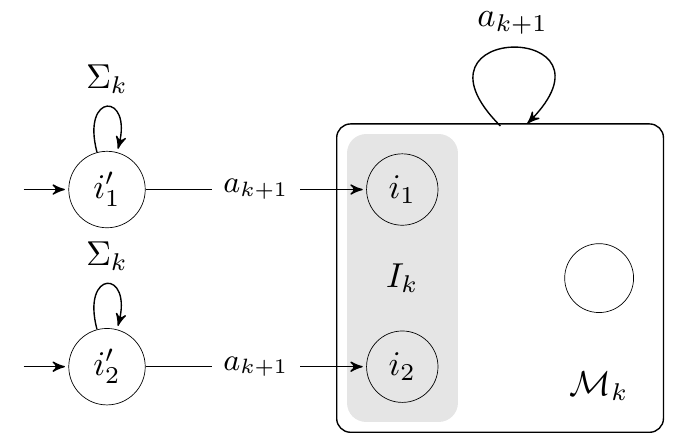}
      \caption{The ptNFA $\M_{k+1}$ constructed from the ptNFA $\M_k$ with two initial states $I_k=\{i_1,i_2\}$}
      \label{reduction}
    \end{figure}
    We now prove that $L_k$ is $k$-piecewise testable if and only if $L_{k+1}$ is $(k+1)$-piecewise testable.
    
    Assume that $L_k$ is $k$-piecewise testable. Let $x,y\in \Sigma_{k+1}^*$ be two words such that $x\sim_{k+1} y$. Since $k+1\ge 1$, we have that $\alp(x) = \alp(y)$. If $a_{k+1} \notin \alp(x)$, then neither $x$ nor $y$ belongs to $L_{k+1}$. Thus, assume that $a_{k+1}$ appears in $x$ and $y$. Then $x=x' a_{k+1} x''$ and $y=y' a_{k+1} y''$, where $a_{k+1}\notin \alp(x'y')$. By Lemma~\ref{lemma2}, $x'' \sim_k y''$. By construction, the words $x''$ and $y''$ are read in $\M_k$ extended with the self-loops under $a_{k+1}$. Let $p:\Sigma_{k+1}^* \to \Sigma_k^*$ denote a morphism such that $p(a_{k+1})=\eps$ and $p(a)=a$ for every $a\in\Sigma_k$. Since no $a_{k+1}$-transition changes the state in any computation of $\M_k$, the sets of states reachable by $x$ and $y$ in $\M_{k+1}$ are exactly those reachable by $p(x'')$ and $p(y'')$ in $\M_k$. Since $L_k$ is $k$-piecewise testable, either both contain an accepting state or neither does. Hence $x$ is accepted if and only if $y$ is accepted, which shows that $L_{k+1}$ is $(k+1)$-piecewise testable.
    
    On the other hand, assume that $L_k$ is not $k$-piecewise testable. Then there exist words $x$ and $y$ such that $x \sim_k y$ and $|L_k\cap\{x,y\}| = 1$. Let $w\in\Sigma_k^*$ be such that $\fun{sub}_{k+1}(w)=\{ u\in\Sigma_k^* \mid |u| \le k+1\}$. Then, $w a_{k+1} x \sim_{k+1} w a_{k+1} y$ and $|L_{k+1} \cap \{w a_{k+1} x, w a_{k+1} y\}| = 1$. This shows that $L_{k+1}$ is not $(k+1)$-piecewise testable.
  \end{proof}

  \begin{remark}[Parallel composition]
    A morphism $p:\Sigma^* \to \Sigma_o^*$, for $\Sigma_o\subseteq \Sigma$, defined as $p(a)=a$, for $a\in\Sigma_o$, and $p(a)=\eps$, otherwise, is called a (natural) projection. Arguments similar to those used in the proof of Lemma~\ref{lemKtoK+1} show that piecewise testable languages are closed under inverse projection. A parallel composition of languages $(L_i)_{i=1}^{n}$ over the alphabets $(\Sigma_i)_{i=1}^{n}$ is defined as $\parallel_{i=1}^{n} L_i = \bigcap_{i=1}^{n} p^{-1}(L_i)$, where $p: (\bigcup_{i=1}^{n} \Sigma_i)^* \to \Sigma_i^*$ is a natural projection. As a consequence, piecewise testable languages are closed under parallel composition. On the other hand, note that piecewise testable languages are not closed under natural projection.
  \end{remark}
  
  \begin{reptheorem}{theorem27}
    Let $\Sigma$ be a fixed alphabet with $c=|\Sigma|\ge 2$, and let $k\ge 0$. Then the problem to decide whether the language of a ptNFA $\A$ over $\Sigma$ is $k$-piecewise testable is coNP-complete.
  \end{reptheorem}
  \begin{proof}
    Let $d$ denote the depth of $\A$. Then the language $L(\A)$ is $d$-piecewise testable. If $k \ge d$, then the answer is {\sc Yes}. Thus, assume that $k < d$. Notice that if $u \sim_d v$, then $u \sim_k v$, but the opposite does not hold. If $L(\A)$ is not $k$-piecewise testable, then there exist two words $x\in L(\A)$ and $y\notin L(\A)$ such that $x \sim_k y$. This means that $x \not\sim_d y$, hence we can guess the minimal representatives of the $x/_{\sim_d}$ and $y/_{\sim_d}$ classes that are of length $O(d^c)$, see the discussion above, which is polynomial in the depth of $\A$, and check that $x\in L(\A)$ and $y\notin L(\A)$, and that $x\sim_k y$. The last step requires to test all words up to length $k$ for embedding in both words. However, it is at most $kc^k$ words, which is a constant.
    
    To prove hardness, we reduce 0-piecewise testability to $k$-piecewise testability, $k\ge 1$. First, notice that the proof of Lemma~\ref{lemKtoK+1} cannot be used, since the alphabet there grows proportionally to $k$. However, the proof here is a simple modification of that proof. Let $\M_0$ over $\Sigma_0$ be a ptNFA. Construct the ptNFA $\M_{k}$ over the alphabet $\Sigma=\Sigma_0 \cup \{a\}$, where $a\notin\Sigma_0$, as depicted in Figure~\ref{reduction2}. Namely, $\M_{k}$ is constructed from $\M_0$ by adding self-loops under $a$ to every state of $\M_0$, and by adding $k$ new states $i_{j,1},\ldots,i_{j,k}$ for every initial state $i_j$ of $\M_k$. Every $i_{j,\ell}$ contains self-loops under all letters from $\Sigma_0$ and $i_{j,\ell}$ goes to $i_{j,\ell+1}$ under $a$, for $1\le \ell < k-1$, and $i_{j,k}$ goes to the initial states $i_j$ of $\M_0$ under $a$. The initial states of $\M_{k}$ are the states $i_{j,1}$, the accepting states are the accepting states of $\M_0$. Note that $\M_k$ is a ptNFA.
    \begin{figure}
      \centering
      \includegraphics{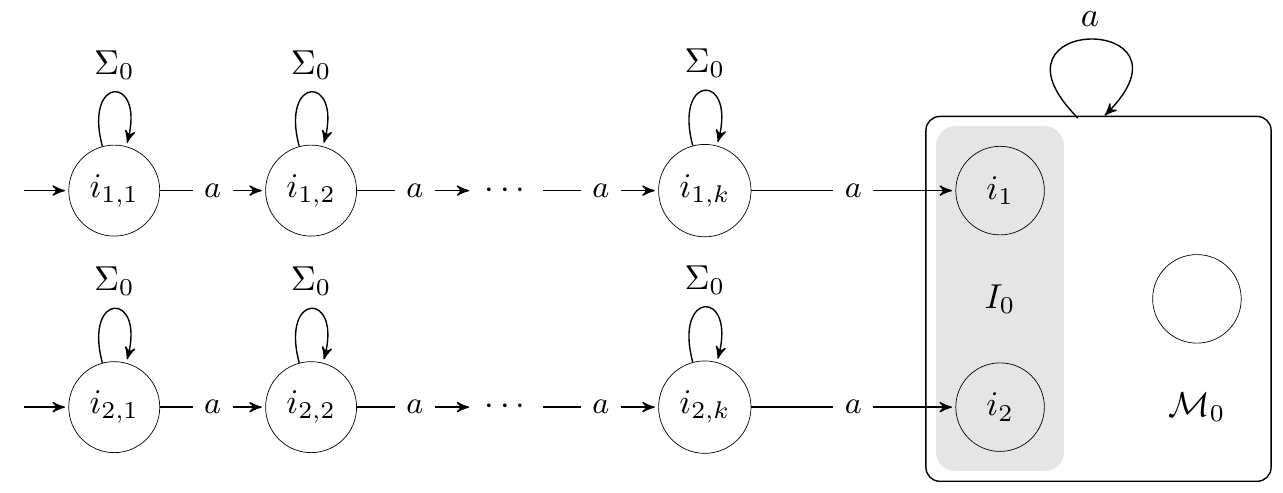}
      \caption{The ptNFA $\M_{k}$ constructed from a ptNFA $\M_0$ with two initial states}
      \label{reduction2}
    \end{figure}
    We now prove that $L(\M_k)$ is $k$-piecewise testable if and only if $L(\M_{0})$ is $0$-piecewise testable.
    
    Assume that $L(\M_0)$ is $0$-piecewise testable. Let $x,y\in \Sigma^*$ be two words such that $x\sim_{k} y$. If $a^k \not\preccurlyeq x$, then $x \notin L(\M_k)$, and $x\sim_k y$ implies that $a^k \not\preccurlyeq y$, hence $y \notin L(\M_{k})$ either. Thus, assume that $a^k\preccurlyeq x$ and $a^k\preccurlyeq y$. Then $x=x_1ax_2a\ldots x_k a x''$ and $y=y_1 a y_2 a \ldots y_k a y''$, where $a\notin \alp(x_1\cdots x_ky_1\cdots y_k)$. By Lemma~\ref{lemma2} applied $k$-times, $x'' \sim_0 y''$. Notice that, by construction, the words $x''$ and $y''$ are read in $\M_0$ extended with the self-loops under $a$ and the sets of states reachable by $x$ and $y$ in $\M_{k}$ are exactly those reachable by $x''$ and $y''$ in $\M_0$. Let $p:\Sigma^* \to \Sigma_0^*$ denote a morphism such that $p(a)=\eps$ and $p(b)=b$ for $b\in\Sigma_0$. Since no $a$-transition changes the state in any computation of $\M_0$, the sets of states reachable by $x''$ and $y''$ in $\M_0$ are exactly those reachable by $p(x'')$ and $p(y'')$, respectively. Since $L(\M_0)$ is $0$-piecewise testable, either both contain an accepting state or neither does. Together, $x$ is accepted by $\M_k$ if and only if $y$ is accepted by $\M_k$, which shows that $L(\M_{k})$ is $k$-piecewise testable.
    
    On the other hand, assume that $L(\M_0)$ is not $0$-piecewise testable. Then there are two words $x \in L(\M_0)$ and $y\notin L(\M_0)$. Let $w\in\Sigma_0^*$ be such that $\fun{sub}_{k}(w)=\{ u\in\Sigma_0^* \mid |u| \le k\}$. Then, we have that $(w a)^k x \sim_{k} (wa)^k y$ and $|L(\M_{k}) \cap \{ (wa)^{k} x, (wa)^{k} y\}| = 1$. This shows that $L(\M_{k})$ is not $k$-piecewise testable.
  \end{proof}

  \begin{repcorollary}{cor28}
    The $k$-piecewise testability problem for ptNFAs over an alphabet $\Sigma$ is coNP-hard for $k\ge 0$ even if $|\Sigma|=3$.
  \end{repcorollary}
  \begin{proof}
    It is shown in Lemma~\ref{0ptNFAhard} that 0-piecewise testability for ptNFAs is coNP-hard for a binary alphabet. The hardness proof of the previous theorem then shows that, for any $k\ge 1$, $k$-piecewise testability is coNP-hard for a ternary alphabet.
  \end{proof}

  \begin{reptheorem}{thmP}
    The $k$-piecewise testability problem for ptNFAs over a unary alphabet is decidable in polynomial time. The result holds even if $k$ is part of the input.
  \end{reptheorem}
  \begin{proof}
    Let $\A$ be a ptNFA of depth $d$. Then the language $L(\A)$ is $d$-piecewise testable by Theorem~\ref{thmMain} and the minimal representatives of $\sim_d$-classes are of length at most $d$; there are at most $d+1$ equivalence classes. If $k\ge d$, then the language $L(\A)$ is $k$-piecewise testable, since every $d$-piecewise testable language is also $(d+1)$-piecewise testable. If $k < d$, then the language $L(\A)$ is not $k$-piecewise testable if and only if there are two words of length at most $d$ that are $\sim_k$-equivalent and only one of them is accepted. Since all words of length less than $k$ are $\sim_k$-equivalent only with itself and all unary words of length at least $k$ are $\sim_k$-equivalent, it can be checked in polynomial time whether there is a word of length at least $k+1$ and at most $d$ with a different acceptance status than $a^k$.
  \end{proof}

  \begin{reptheorem}{thm30}
    Both piecewise testability and $k$-piecewise testability problems for NFAs over a unary alphabet are coNP-complete.
  \end{reptheorem}
  \begin{proof}
    We first show that to check piecewise testability for NFAs over a unary alphabet is in coNP. To do this, we show how to check non-piecewise testability in NP. By Fact~\ref{thm:characterization}, we need to check that the corresponding DFA is partially ordered and confluent. However, confluence is trivially satisfied because there is no branching in a DFA over a single letter. Partial order is violated if and only if there exist three words $a^{\ell_1}$, $a^{\ell_2}$ and $a^{\ell_3}$ with $\ell_1 < \ell_2 < \ell_3$ such that $I\cdot a^{\ell_1} = I\cdot a^{\ell_3} \neq I \cdot a^{\ell_2}$ and one of these sets is accepting and the other is not (otherwise they are equivalent). The lengths are bounded by $2^n$, where $n$ denotes the number of states of the NFA, and can be guessed in binary. The fast matrix multiplication can then be used to compute resulting sets of those transitions in polynomial time.
    
    Thus, we can check in coNP whether the language of an NFA is piecewise testable. If so, then it is $2^n$-piecewise testable, since the depth of the minimal DFA is bounded by $2^n$, where $n$ is the number of states of the NFA. Let $M$ be the transition matrix of the NFA. To show that it is not $k$-piecewise testable, we need to find two $\sim_k$-equivalent words such that exactly one of them belongs to the language of the NFA. Since every class defined by $a^\ell$, for $\ell < k$, is a singleton, we need to find $k< \ell \le 2^n$ such that $a^k \sim_k a^\ell$ and only one of them belongs to the language. This can be done in nondeterministic polynomial time by guessing $\ell$ in binary and using the matrix multiplication to obtain the reachable sets in $M^k$ and $M^\ell$ and verifying that one is accepting and the other is not.

    We now show that both problems are coNP-hard. To do this, we use the proof given in~\cite{StockmeyerM73} showing that universality is coNP-hard. We recall it here for convenience. 
    
    Let $\varphi$ be a formula in 3CNF with $n$ distinct variables, and let $C_k$ be the set of literals in the $k$th conjunct, $1 \le k \le m$. The assignment to the variables can be represented as a binary vector of length $n$. Let $p_1,p_2,\ldots,p_n$ be the first $n$ prime numbers. For a natural number $z$ congruent with 0 or 1 modulo $p_i$, for every $1\le i \le n$, we say that $z$ satisfies $\varphi$ if the assignment $(z \bmod p_1, z \bmod p_2,\ldots, z \bmod p_n)$ satisfies $\varphi$. Let 
    \[
      E_0 = \bigcup_{k=1}^{n} \bigcup_{j=2}^{p_k-1} 0^j\cdot (0^{p_k})^*
    \]
    that is, $L(E_0) = \{ 0^z \mid \exists k \le n, z \not\equiv 0 \bmod p_k \text{ and } z \not\equiv 1 \bmod p_k \}$ is the set of natural numbers that do not encode an assignment to the variables. For each conjunct $C_k$, we construct an expression $E_k$ such that if $0^z \in L(E_k)$ and $z$ is an assignment, then $z$ does not assign the value 1 to any literal in $C_k$. For example, if $C_k = \{x_{1,r}, \neg x_{1,s}, x_{1,t}\}$, for $1 \le  r,s,t \le n$ and $r,s,t$ distinct, let $z_k$ be the unique integer such that $0\le z_k < p_rp_sp_t$, $z_k \equiv 0 \bmod p_r$, $z_k \equiv 1 \bmod p_s$, and $z_k \equiv 0 \bmod p_t$. Then
    \[
      E_k = 0^{z_k} \cdot (0^{p_rp_sp_t})^*\,.
    \]
    Now, $\varphi$ is satisfiable if and only if there exists $z$ such that $z$ encodes an assignment to $\varphi$ and $0^z \notin L(E_k)$ for all $1\le k \le m$, which is if and only if $L(E_0 \cup \bigcup_{k=1}^{m} E_k) \neq 0^*$.
    
    The proof shows that universality is coNP-hard for NFAs over a unary alphabet. Let $p_n\# = \Pi_{i=1}^{n} p_i$. If $z$ encodes an assignment of $\varphi$, then, for any natural number $c$, $z+c\cdot p_n\#$ also encodes an assignment of $\varphi$. Indeed, if $z \equiv x_i \bmod p_i$, then $z + c\cdot p_n\# \equiv x_i \bmod p_i$, for every $1\le i\le n$. This shows that if  $0^z \notin L(E_k)$ for all $k$, then $0^z (0^{p_n\#})^* \cap L(E_0 \cup \bigcup_{k=1}^{m} E_k) = \emptyset$. Since both languages are infinite, the minimal DFA recognizing the language $L(E_0 \cup \bigcup_{k=1}^{m} E_k)$ must have a non-trivial cycle. Therefore, if the language is universal, then it is $k$-piecewise testable, for any $k\ge 0$, and if it is non-universal, then it is not piecewise testable. This proves coNP-hardness of $k$-piecewise testability for every $k\ge 0$.
  \end{proof}

\end{document}